\documentclass[11pt]{article}
\usepackage{odonnell}

\newtheorem{conjecture}[theorem]{Conjecture}

\newcommand{\threesat}{\textsf{3Sat}\xspace}
\newcommand{\alin}[1]{\textsf{#1NLin}\xspace}
\newcommand{\notall}{\textsf{4-Not-All-There}\xspace}
\newcommand{\fournat}{\textsf{4NAT}\xspace}
\newcommand{\csp}{\mathsf{CSP}}
\newcommand{\instance}{\mathcal{I}}
\newcommand{\buddy}{\textsf{TwoPair}\xspace}
\newcommand{\twopair}{\buddy}
\newcommand{\gadg}{\gamma}
\newcommand{\roots}{U}
\newcommand{\x}{\bx}
\newcommand{\y}{\by}
\newcommand{\z}{\bz}
\newcommand{\w}{\bw}
\newcommand{\pih}{\pi_3}

\begin{document}

\title{A new point of NP-hardness for 2-to-1 Label Cover}

\author{Per Austrin\thanks{Department of Computer Science, University of Toronto.  Funded by NSERC.} \and Ryan O'Donnell\thanks{Department of Computer Science, Carnegie Mellon University. Supported by NSF grants CCF-0747250 and CCF-0915893, and by a Sloan fellowship.} \and John Wright\thanks{Department of Computer Science, Carnegie Mellon University.}}

\maketitle

\begin{abstract}
    We show that given a satisfiable instance of the $2$-to-$1$ Label Cover problem, it is $\NP$-hard to find a $(\frac{23}{24} + \eps)$-satisfying assignment. 
\end{abstract}

\setcounter{page}{0}
\thispagestyle{empty}
\newpage
\section{Introduction}

Over the past decade, a significant amount of progress has been made in the field of hardness of approximation via results based on the conjectured hardness of certain forms of the Label Cover problem.  The \emph{Unique Games Conjecture} (UGC) of Khot~\cite{Kho02} states that it is $\NP$-hard to distinguish between nearly satisfiable and almost completely unsatisfiable instances of \emph{Unique}, or \emph{1-to-1}, Label Cover. Using the UGC as a starting point, we now have optimal inapproximability results for Vertex Cover~\cite{KR03}, Max-Cut~\cite{KKMO07}, and many other basic constraint satisfaction problems (CSP).  Indeed, assuming the UGC we have essentially optimal inapproximability results for \emph{all} CSPs~\cite{Rag08}.  In short, modulo the understanding of Unique Label Cover itself, we have an excellent understanding of the (in-)approximability of a wide range of problems.

Where the UGC's explanatory powers falter is in pinning down the approximability of \emph{satisfiable} CSPs.  This means the task of finding a good assignment to a CSP when guaranteed that the CSP is fully satisfiable.  For example, we know from the work of \Hastad~\cite{Has01} that given a fully satisfiable $\threesat$ instance, it is $\NP$-hard to satisfy $\frac78 + \eps$ of the clauses for any $\eps > 0$.  However given a fully satisfiable $1$-to-$1$ Label Cover instance, it is completely trivial to find a fully satisfying assignment.  Thus the UGC can not be used as the starting point for hardness results for satisfiable CSPs.  Because of this, Khot additionally posed his $d$-to-$1$ Conjectures:
\begin{conjecture}[\cite{Kho02}]
For every integer $d \geq 2$ and $\eps > 0$, there is a label set size~$q$ such that it is $\NP$-hard to $(1, \eps)$-decide the $d$-to-$1$ Label Cover problem.
\end{conjecture}
Here by $(c, s)$-deciding a CSP we mean the task of determining whether an instance is at least $c$-satisfiable or less than $s$-satisfiable.   It is well known (from the Parallel Repetition Theorem~\cite{FK94,Raz95}) that the conjecture is true if~$d$ is allowed to depend on~$\eps$.  The strength of this conjecture, therefore, is that it is stated for each fixed $d$ greater than~$1$.

The $d$-to-$1$ Conjectures have been used to resolve the approximability of several basic ``satisfiable CSP'' problems.  The first result along these lines was due to Dinur, Mossel, and Regev~\cite{DMR09} who showed that the $2$-to-$1$ Conjecture implies that it is $\NP$-hard to $C$-color a $4$-colorable graph for any constant~$C$. (They also showed hardness for $3$-colorable graphs via another Unique Games variant.)  O'Donnell and Wu~\cite{OW09b} showed that assuming the $d$-to-$1$ Conjecture for any fixed~$d$ implies that it is $\NP$-hard to $(1,\frac58 + \eps)$-approximate instances a certain $3$-bit predicate --- the ``Not Two'' predicate.  This is an optimal result among all $3$-bit predicates, since Zwick~\cite{Zwi98} showed that every satisfiable $3$-bit CSP instance can be efficiently $\frac58$-approximated.  In another example, Guruswami and Sinop~\cite{GS09} have shown that the $2$-to-$1$ Conjecture implies that given a $q$-colorable graph, it is $\NP$-hard to find a $q$-coloring in which less than a $(\frac{1}{q} - O(\frac{\ln q}{q^2}))$ fraction of the edges are monochromatic.  This result would be tight up to the $O(\cdot)$ by an algorithm of Frieze and Jerrum~\cite{FJ97}. It is therefore clear that settling the $d$-to-1 Conjectures, especially in the most basic case of $d=2$, is an important open problem.

Regarding the hardness of the $2$-to-$1$ Label Cover problem, the only evidence we have is a family of integrality gaps for the canonical SDP relaxation of the problem, in~\cite{GKO+10}.  Regarding algorithms for the problem, an important recent line of work beginning in~\cite{ABS10} (see also~\cite{BRS11,GS11,Ste10}) has sought subexponential-time algorithms for Unique Label Cover and related problems.  In particular, Steurer~\cite{Ste10} has shown that for any constant $\beta > 0$ and label set size, there is an $\exp(O(n^{\beta}))$-time algorithm which, given a satisfiable $2$-to-$1$ Label Cover instance, finds an assignment satisfying an $\exp(-O(1/\beta^2))$-fraction of the constraints.  E.g., there is a $2^{O(n^{.001})}$-time algorithm which $(1,s_0)$-approximates $2$-to-$1$ Label Cover, where $s_0 > 0$ is a certain universal constant.

In light of this, it is interesting not only to seek $\NP$-hardness results for certain approximation thresholds, but to additionally seek evidence that \emph{nearly full exponential time} is required for these thresholds.  This can done by assuming the Exponential Time Hypothesis (ETH)~\cite{IP01} and by reducing from the Moshkovitz--Raz Theorem~\cite{MR10}, which shows a near-linear size reduction from \threesat to the standard Label Cover problem with subconstant soundness.  In this work, we show reductions from $\threesat$  to the problem of $(1, s+\eps)$-approximating several CSPs, for certain values of $s$ and for all $\eps >0$.  In fact, though we omit it in our theorem statements, it can be checked that all of the reductions in this paper are quasilinear in size for $\eps = \eps(n) = \Theta\left(\frac{1}{(\log \log n)^\beta}\right)$, for some $\beta > 0$.


\subsection{Our results}     \label{sec:our-results}
In this paper, we focus on proving $\NP$-hardness for the $2$-to-$1$ Label Cover problem.  To the best of our knowledge, no explicit $\NP$-hardness factor has previously been stated in the literature.  However it is ``folklore'' that one can obtain an explicit one for label set sizes $3$ \& $6$ by performing the ``constraint-variable'' reduction on an $\NP$-hardness result for $3$-coloring (more precisely, Max-$3$-Colorable-Subgraph).  The best known hardness for $3$-coloring is due to Guruswami and Sinop~$\cite{GS09}$, who showed a factor $\frac{32}{33}$-hardness via a somewhat involved gadget reduction from the $3$-query adaptive PCP result of~\cite{GLST98}.  This yields $\NP$-hardness of $(1, \frac{65}{66}+ \eps)$-approximating $2$-to-$1$ Label Cover with label set sizes $3$ \& $6$.  It is not known how to take advantage of larger label set sizes.  On the other hand, for label set sizes $2$~\&~$4$ it is known that satisfying $2$-to-$1$ Label Cover instances can be found in polynomial time.

The main result of our paper gives an improved hardness result:
\begin{theorem}\label{thm:twotoone-hardness}
For all $\eps > 0$, $(1, \frac{23}{24} + \eps)$-deciding the $2$-to-$1$ Label Cover problem with label set sizes $3$ \textnormal{\&} $6$ is $\NP$-hard.  
\end{theorem}
\noindent By duplicating labels, this result also holds for label set sizes $3k$ \& $6k$ for any $k \in \N^+$.\\

Let us describe the high-level idea behind our result.  The folklore constraint-variable reduction from $3$-coloring to $2$-to-$1$ Label Cover would work just as well if we started from ``$3$-coloring with literals'' instead.  By this we mean the CSP with domain $\Z_3$ and constraints of the form ``$v_i - v_j \neq c \pmod{3}$''.  Starting from this CSP --- which we call $\alin{2}(\Z_3)$ --- has two benefits: first, it is at least as hard as $3$-coloring and hence could yield a stronger hardness result; second, it is a bit more ``symmetrical'' for the purposes of designing reductions.  We obtain the following hardness result for $\alin{2}(\Z_3)$.
\begin{theorem}\label{thm:alin-hardness}
For all $\eps > 0$, it is $\NP$-hard to $(1, \frac{11}{12} + \eps)$-decide the $\alin{2}$ problem.
\end{theorem}
\noindent
As 3-coloring is a special case of $\alin{2}(\Z_3)$, $\cite{GS09}$ also shows that $(1, \frac{32}{33}+\eps)$-deciding $\alin{2}$ is $\NP$-hard for all $\eps > 0$, and to our knowledge this was previously the only hardness known  for $\alin{2}(\Z_3)$. The best current algorithm achieves an approximation ratio of $0.836$ (and does not need the instance to be satisfiable)~\cite{GW04}.
To prove Theorem~\ref{thm:alin-hardness}, we proceed by designing an appropriate ``function-in-the-middle'' dictator test, as in the recent framework of~\cite{OW12}.  Although the~\cite{OW12} framework gives a direct translation of certain types of function-in-the-middle tests into hardness results, we cannot employ it in a black-box fashion.  Among other reasons,~\cite{OW12} assumes that the test has ``built-in noise'', but we cannot afford this as we need our test to have perfect completeness.

Thus, we need a different proof to derive a hardness result from this function-in-the-middle test.  We first were able to accomplish this by an analysis similar to the Fourier-based proof of $2\mathsf{Lin}(\Z_2)$ hardness given in Appendix~F of~\cite{OW12}.  Just as that proof ``reveals'' that the function-in-the-middle $2\mathsf{Lin}(\Z_2)$ test can be equivalently thought of as \Hastad's $3\mathsf{Lin}(\Z_2)$ test composed with the $3\mathsf{Lin}(\Z_2)$-to-$2\mathsf{Lin}(\Z_2)$ gadget of~\cite{TSSW00}, our proof for the $\alin{2}(\Z_3)$ function-in-the-middle test revealed it to be the composition of a function test for a certain four-variable CSP with a gadget.  We have called the particular four-variable CSP \notall, or $\fournat$ for short.  Because it is a $4$-CSP, we are able to prove the following $\NP$-hardness of approximation result for it using a classic, \Hastad-style Fourier-analytic proof.
\begin{theorem}\label{thm:fournat-hardness}
For all $\eps > 0$, it is $\NP$-hard to $(1, \frac{2}{3} + \eps)$-decide the $\fournat$ problem.
\end{theorem}
\noindent
Thus, the final form in which we present our Theorem~\ref{thm:twotoone-hardness} is as a reduction from Label-Cover to $\fournat$ using a function test (yielding Theorem~\ref{thm:fournat-hardness}), followed by a $\fournat$-to-$\alin{2}(\Z_3)$ gadget (yielding Theorem~\ref{thm:alin-hardness}), followed by the constraint-variable reduction to $2$-to-$1$ Label Cover.  Indeed, all of the technology needed to carry out this proof was in place for over a decade, but without the function-in-the-middle framework of~\cite{OW12} it seems that pinpointing the $\fournat$ predicate as a good starting point would have been unlikely.

\subsection{Organization}


We leave to Section~\ref{sec:prelims} most of the definitions, including those of the CSPs we use.  The heart of the paper is in Section~\ref{sec:better}, where we give both the $\alin{2}(\Z_3)$ and $\fournat$ function tests, explain how one is derived from the other, and then perform the Fourier analysis for the $\fournat$ test.  The actual hardness proof for $\fournat$ is presented in Section~\ref{app:fournat}, and it follows mostly the techniques put in place by \Hastad in \cite{Has01}.  

\section{Preliminaries}\label{sec:prelims}

We primarily work with strings $x \in \Z_3^K$ for some integer $K$.  We write $x_i$ to denote the $i$th coordinate of~$x$.  Oftentimes, our strings $y \in \Z_3^{dK}$ are ``blocked'' into $K$ ``blocks'' of size $d$.  In this case, we write $y[i] \in \Z_3^{d}$ for the $i$th block of~$y$, and $(y[i])_j \in \Z_3$ for the $j$th coordinate of this block.  Define the function $\pi:[dK]\rightarrow [K]$ such that $\pi(k) = i$ if $k$ falls in the $i$th block of size $d$ (e.g., $\pi(k) = 1$ for $1 \leq k \leq d$, $\pi(k) = 2$ for $d+1 \leq k \leq 2d$, and so on).

\subsection{Definitions of problems}

An instance $\instance$ of a \emph{constraint satisfaction problem} (CSP) is a set of variables $V$, a set of labels $D$, and a weighted list of constraints on these variables.  We assume that the weights of the constraints are nonegative and sum to 1.  The weights therefore induce a probability distribution on the constraints.  Given an assignment to the variables $f:V\rightarrow D$, the \emph{value} of $f$ is the probability that $f$ satisfies a constraint drawn from this probability distribution.  The \emph{optimum} of $\instance$ is the highest value of any assignment.  We say that an $\instance$ is $s$-\emph{satisfiable} if its optimum is at least $s$.  If it is 1-satisfiable we simply call it satisfiable.

We define a CSP $\calP$ to be a set of CSP instances.  Typically, these instances will have similar constraints.  We will study the problem of \emph{$(c, s)$-deciding} $\calP$.  This is the problem of determining whether an instance of $\calP$ is at least $c$-satisfiable or less than $s$-satisfiable.  Related is the problem of \emph{$(c, s)$-approximating} $\calP$, in which one is given a $c$-satisfiable instance of $\calP$ and asked to find an assignment of value at least $s$.  It is easy to see that $(c, s)$-deciding $\calP$ is at least as easy as $(c, s)$-approximating $\calP$.  Thus, as all our hardness results are for $(c, s)$-deciding CSPs, we also prove hardness for $(c, s)$-approximating these CSPs.

We now state the three CSPs that are the focus of our paper.
\paragraph{\textsf{2-NLin($\Z_3$)}:} In this CSP the label set is $\Z_3$ and the constraints are of the form
\begin{equation*}
v_i - v_j \neq a \pmod{3}, \quad a \in \Z_3.
\end{equation*}
The special case when each RHS is~$0$ is the $3$-coloring problem.  We often drop the $(\Z_3)$ from this notation and simply write $\alin{2}$. The reader may think of the `\textsf{N}' in $\alin{2}(\Z_3)$ as standing for `N'on-linear, although we prefer to think of it as standing for `N'early-linear.  The reason is that when generalizing to moduli $q > 3$, the techniques in this paper generalize to constraints of the form ``$v_i  - v_j \pmod{q} \in \{a, a+1\}$'' rather than ``$v_i - v_j \neq a \pmod{q}$''.  For the ternary version of this constraint, ``$v_i - v_j + v_k \pmod{q} \in \{a,a+1\}$'', it is folklore\footnote{Venkatesan Guruswami, Subhash Khot personal communications.} that a simple modification of \Hastad's work~\cite{Has01} yields $\NP$-hardness of $(1,\frac{2}{q})$-approximation.

\paragraph{\notall:}  For the \notall problem, denoted $\fournat$, we define $\fournat \co \Z_3^4 \rightarrow \{0, 1\}$ to have output $1$ if and only if at least one of the elements of $\Z_3$ is not present among the four inputs.  The $\fournat$ CSP has label set $D = \Z_3$ and constraints of the form $\fournat(v_1 + k_1, v_2 + k_2, v_3 + k_3, v_4 + k_4) = 1$, where the $k_i$'s are constants in $\Z_3$.

We additionally define the ``\textsf{Two Pairs}'' predicate $\buddy \co \Z_3^4 \rightarrow \{0, 1\}$, which has output $1$ if and only if its input contains two distinct elements of $\Z_3$, each appearing twice.  Note that an input which satisfies $\buddy$ also satisfies $\fournat$.

\paragraph{$\mathbf{d}$-to-1 Label Cover:}  An instance of the $d$-to-1 Label Cover problem is a bipartite graph $G=(U \cup V, E)$, a label set size $K$, and a $d$-to-1 map $\pi_e:[dK] \rightarrow [K]$ for each edge $e \in E$.  The elements of $U$ are labeled from the set $[K]$, and the elements of $V$ are labeled from the set $[dK]$.  A labeling $f: U \cup V \rightarrow [dK]$ satisfies an edge $e = (u, v)$ if $\pi_e(f(v)) = f(u)$.  Of particular interest is the $d = 2$ case, i.e., 2-to-1 Label Cover.

Label Cover serves as the starting point for most $\NP$-hardness of approximation results.  We use the following theorem of Moshkovitz and Raz:
\begin{theorem}[\cite{MR10}]\label{thm:mosraz}
For any $\eps = \eps(n) \geq n^{-o(1)}$ there exists $K, d \leq 2^{\mathrm{poly}(1/\eps)}$ such that the problem of deciding a \threesat instance of size $n$ can be Karp-reduced in $\mathrm{poly}(n)$ time to the problem of $(1, \eps)$-deciding  $d$-to-1 Label Cover instance of size $n^{1+o(1)}$ with label set size~$K$.
\end{theorem}

\subsection{Gadgets}

A typical way of relating two separate CSPs is by constructing a \emph{gadget reduction} which translates from one to the other.  A gadget reduction from $\csp_1$ to $\csp_2$ is one which maps any $\csp_1$ constraint into a weighted set of $\csp_2$ constraints.  The $\csp_2$ constraints are over the same set of variables as the $\csp_1$ constraint, plus some new, auxiliary variables (these auxiliary variables are not shared between constraints of $\csp_1$).  We require that for every assignment which satisfies the $\csp_1$ constraint, there is a way to label the auxiliary variables to fully satisfy the $\csp_2$ constraints.  Furthermore, there is some parameter $0<\gadg <1$ such that for every assignment which does not satisfy the $\csp_1$ constraint, the optimum labeling to the auxiliary variables will satisfy exactly $\gadg$ fraction of the $\csp_2$ constraints.  Such a gadget reduction we call a \emph{$\gadg$-gadget-reduction} from $\csp_1$ to $\csp_2$.  The following proposition is well-known:
\begin{proposition}\label{prop:gadgets}
Suppose it is $\NP$-hard to $(c, s)$-decide $\csp_1$.  If there exists a $\gadg$-gadget-reduction from $\csp_1$ to $\csp_2$, then it is $\NP$-hard to $(c+(1-c)\gadg, s + (1-s)\gadg)$-decide $\csp_2$.
\end{proposition}
We note that the notation $\gadg$-gadget-reduction is similar to a piece of notation employed by \cite{TSSW00}, but the two have different (though related) definitions.

\subsection{Fourier analysis on $\Z_3$}\label{sec:fourier}

Let $\omega = e^{2\pi i/3}$ and set $\roots_3 = \{\omega^0, \omega^1, \omega^2\}$.  For $\alpha \in \Z_3^n$, consider the Fourier character $\chi_\alpha :\Z_3^n \rightarrow \roots_3$ defined as $\chi_\alpha(x) = \omega^{\alpha\cdot x}$.  Then it is easy to see that $\E[\chi_{\alpha}(\x)\overline{\chi_{\beta}(\x)}] = {\bf 1}[\alpha = \beta]$, where here and throughout $\bx$ has the uniform probability distribution on $\Z_3^n$ unless otherwise specified..  As a result, the Fourier characters form an orthonormal basis for the set of functions $f:\Z_3^n \rightarrow \roots_3$ under the inner product $\la f,g \ra = \E[f(\x) g(\x)]$; i.e.,
\begin{equation*}
f = \sum_{\alpha \in \Z_3^n} \hat{f}(\alpha)\chi_\alpha,
\end{equation*}
where the $\hat{f}(\alpha)$'s are complex numbers defined as $\hat{f}(\alpha) = \E[f(\x)\overline{\chi_\alpha(\x)}]$.  For $\alpha \in \Z_3^n$, we use the notation $\vert \alpha \vert$ to denote $\sum \alpha_i$ and $\#\alpha$ to denote the number of nonzero coordinates in $\alpha$.  When $d$ is clear from context and $\alpha \in \Z_3^{dK}$, define $\pih(\alpha) \in \Z_3^K$ so that $(\pih(\alpha))_i \equiv \vert \alpha[i] \vert \pmod{3}$ (recall the notation $\alpha[i]$ from the beginning of this section).

We have Parseval's identity: for
every $f:\Z_3^n \rightarrow \roots_3$ it holds that $\sum_{\alpha\in\Z_3^n}\vert\hat{f}(\alpha)\vert^2 = 1$.  Note that this implies that $\vert \hat{f}(\alpha) \vert \leq 1$ for all $\alpha$, as otherwise $\hat{f}(\alpha)^2$ would be greater than 1.
A function $f:\Z_3^n\rightarrow \Z_3$ is said to be \emph{folded} if for every $x \in \Z_3^n$ and $c \in \Z_3$, it holds that $f(x + c) = f(x) +c$, where $(x+c)_i = x_i + c$.
\begin{proposition}
Let $f:\Z_3^n \rightarrow \roots_3$ be folded.  Then $\hat{f}(\alpha)\neq 0 \Rightarrow \vert \alpha \vert \equiv 1 \pmod{3}$.
\end{proposition}
\begin{proof}
\begin{equation*}
\hat{f}(\alpha) = \E[f(\x + 1)\overline{\chi_\alpha(\x+1)}]
=\E[\omega f(\x)\overline{\chi_\alpha(\x)}\overline{\chi_\alpha(1, 1, \ldots, 1)}]
= \omega \overline{\chi_\alpha(1, 1, \ldots, 1)}\hat{f}(\alpha).
\end{equation*}
This means that $\omega\overline{\chi_\alpha(1, 1, \ldots, 1)}$ must be 1.  Expanding this quantity,
\begin{equation*}
\omega\overline{\chi_\alpha(1, 1, \ldots, 1)} = \omega^{1 - \alpha \cdot (1, 1, \ldots, 1)} = \omega^{1 - \vert \alpha \vert}.
\end{equation*}
So, $\vert \alpha \vert \equiv 1 \pmod{3}$, as promised.
\end{proof}

\section{2-to-1 hardness}\label{sec:better}

In this section, we give our hardness result for 2-to-1 Label Cover, following the proof outline described at the end of Section~\ref{sec:our-results}.
\begin{named}{Theorem~\ref{thm:twotoone-hardness} (restated)}
For all $\eps > 0$, it is $\NP$-hard to $(1, \frac{23}{24} + \eps)$-decide the $2$-to-$1$ Label Cover problem.
\end{named}

First, we state a pair of simple gadget reductions:

\begin{lemma}\label{lem:fournat-alin}
There is a $3/4$-gadget-reduction from $\fournat$ to $\alin{2}$.
\end{lemma}

\begin{lemma}\label{lem:alin-twotoone}
There is a $1/2$-gadget-reduction from $\alin{2}$ to $2$-to-$1$.
\end{lemma}
\noindent
Together with Proposition~\ref{prop:gadgets}, these imply the following corollary:
\begin{corollary}\label{cor:fournat-twotoone}
There is a $7/8$-gadget-reduction from $\fournat$ to $2$-to-$1$.  Thus, if it is $\NP$-hard to $(c, s)$-decide the $\fournat$ problem, then it is $\NP$-hard to $((7+c)/8, (7+s)/8)$-decide the 2-to-1 Label Cover problem.
\end{corollary}
\noindent
The gadget reduction from $\fournat$ to $\alin{2}$ relies on the simple fact that if $a, b, c, d \in \Z_3$ satisfy the $\fournat$ predicate, then there is some element of $\Z_3$ that none of them equal. 

\begin{proof}[Proof of Lemma~\ref{lem:fournat-alin}]
A $\fournat$ constraint $C$ on the variables $S = (v_1, v_2, v_3, v_4)$ is of the form
\begin{equation*}
\fournat(v_1 + k_1, v_2 + k_2, v_3 + k_3, v_4 + k_4),
\end{equation*}
where the $k_i$'s are all constants in $\Z_3$.  To create the $\alin{2}$ instance, introduce the auxiliary variable $y_C$ and add the four $\alin{2}$ equations
\begin{equation}
\label{eqn:fournat-alin}
v_i + k_i \neq y_C \pmod{3}, \quad i \in [4].
\end{equation}

If $f:S \rightarrow \Z_3$ is an assignment which satisfies the $\fournat$ constraint, then there is some $a \in \Z_3$ such that $f(v_i) + k_i \neq a \pmod{3}$ for all $i \in [4]$.  Assigning $a$ to $y_C$ satisfies all four equations \eqref{eqn:fournat-alin}.  On the other hand, if $f$ doesn't satisfy the $\fournat$ constraint, then $\{f(v_i) + k_i\}_{i \in [4]} = \Z_3$, so no assignment to $y_C$ satisfies all four equations.  However, it is easy to see that there is an assignment which satisfies three of the equations.  This gives a $\frac{3}{4}$-gadget-reduction from $\fournat$ to $\alin{2}$, which proves the lemma.
\end{proof}

The reduction from $\alin{2}$ to 2-to-1 Label Cover is the well-known constraint-variable reduction, and uses the fact that in the equation $v_i -  v_j \neq a \pmod{3}$, for any assignment to $v_j$ there are two valid assignments to $v_i$, and vice versa.

\begin{proof}[Proof of Lemma~\ref{lem:alin-twotoone}]
An $\alin{2}$ constraint $C$ on the variables $S = (v_1, v_2)$ is of the form
\begin{equation*}
v_1 - v_2 \neq a \pmod{3},
\end{equation*}
for some $a \in \Z_3$.  To create the 2-to-1 Label Cover instance, introduce the variable $y_C$ which will be labeled by one of the six possible functions $g:S\rightarrow \Z_3$ which satisfies $C$.  Finally, introduce the 2-to-1 constraints $y_C(v_1) = f(v_1)$ and $y_C(v_2) = f(v_2)$.

If $f : S \rightarrow \Z_3$ is an assignment which satisfies the $\alin{2}$ constraint, then we label $y_C$ with $f$.  In this case,
\begin{equation*}
y_C(v_i) = f(v_i),\quad i = 1, 2.
\end{equation*}
Thus, both equations are satisfied.  On the other hand, if $f$ does not satisfy the $\alin{2}$ constraint, then any $g$ which $y_C$ is labeled with disagrees with $f$ on at least one of $v_1$ or $v_2$.  It is easy to see, though, that a $g$ can be selected to satisfy one of the two equations.  This gives a $\frac{1}{2}$-gadget-reduction from $\alin{2}$ to 2-to-1, which proves the lemma.
\end{proof}

\subsection{A pair of tests}

Now that we have shown that $\alin{2}$ hardness results translate into 2-to-1 Label Cover hardness results, we present our $\alin{2}$ function test.  Even though we don't directly use it, it helps explain how we were led to consider the $\fournat$ CSP.  Furthermore, the Fourier analysis that we eventually use for the $\fournat$ Test could instead be performed directly on the $\alin{2}$ Test without any direct reference to the $\fournat$ predicate.  The test is:

\begin{center}
\framebox{$\alin{2}$ Test}
\end{center}
\qquad Given folded functions $f : \Z_3^{K} \rightarrow \Z_3$, $g, h:\Z_3^{dK} \rightarrow \Z_3$:
\begin{itemize}

\item Let $\x \in \Z_3^K$ and $\y \in \Z_3^{dK}$ be independent and uniformly random.
\item For each $i \in [K], j \in [d]$, select $(\z[i])_j$ independently and uniformly from the elements of $\Z_3\setminus\{\x_i, (\y[i])_j\}$.
\item With probability $\frac14$, test $f(\x) \neq h(\z)$; with probability $\frac34$, test $g(\y) \neq h(\z)$.
\end{itemize}

\myfig{.75}{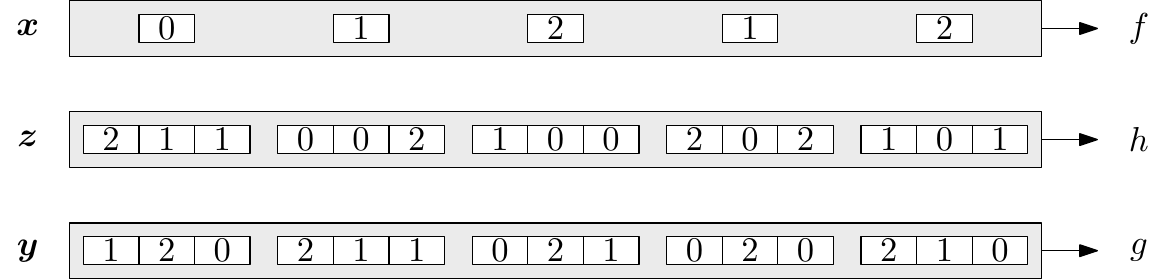}{An illustration of the $\alin{2}$ test distribution; $d = 3$, $K = 5$}{fig:test}

Above is an illustration of the test.  We remark that for any given block $i$, $z[i]$ determines $x_i$ (with very high probability), because as soon as $z[i]$ contains two distinct elements of $\Z_3$, $x_i$ must be the third element of $\Z_3$.  Notice also that in every column of indices, the input to $h$ always differs from the inputs to both $f$ and $g$.  Thus, ``matching dictator'' assignments pass the test with probability~$1$.  (This is the case in which $f(x) = x_i$ and $g(y) = (y[i])_j$ for some $i \in [K]$, $j \in [d]$.) On the other hand, if $f$ and $g$ are ``nonmatching dictators'', then they succeed with only $\frac{11}{12}$ probability.  This turns out to be essentially optimal among functions $f$ and $g$ without ``matching influential coordinates/blocks''.  We will obtain the following theorem:
\begin{named}{Theorem \ref{thm:alin-hardness} restated}
For all $\eps > 0$, it is $\NP$-hard to $(1, \frac{11}{12} + \eps)$-decide the $\alin{2}$ problem.
\end{named}

Before proving this, let us further discuss the $\alin{2}$ test.  Given $\x$, $\y$, and $\z$ from the $\alin{2}$ test, consider the following method of generating two additional strings $\y', \y'' \in \Z_3^{dK}$ which represent $h$'s ``uncertainty'' about $\y$.  For $j \in [d]$, if $\x_i = (\y[i])_j$, then set both $(\y'[i])_j$ and $(\y''[i])_j$ to the lone element of $\Z_3 \setminus \{\x_i, (\z[i])_j\}$.  Otherwise, set one of $(\y'[i])_j$ or $(\y''[i])_j$ to $\x_i$, and the other one to $(\y[i])_j$.  It can be checked that $\buddy(\x_i, (\y[i])_j, (\y'[i])_j, (\y''[i])_j) =  1$, a more stringent requirement than satisfying $\fournat$.  In fact, the marginal distribution on these four variables is a uniformly random assignment that satisfies the $\buddy$ predicate.  Conditioned on $\x$ and $\z$, the distribution on $\y'$ and $\y''$ is identical to the distribution on $\y$.  To see this, first note that by construction, neither $(\y'[i])_j$ nor $(\y''[i])_j$ ever equals $(\z[i])_j$.  Further, because these indices are distributed as uniformly random satisfying assignments to $\buddy$, $\Pr[(\y'[i])_j = x_i] = \Pr[(\y''[i])_j = x_i] = \frac13$, which matches the corresponding probability for $\y$. Thus, as $\y$, $\y'$, and $\y''$ are distributed identically, we may rewrite the test's success probability as:
\begin{align*}
\Pr[\text{$f$, $g$, and $h$ pass the test}]
& = \tfrac14\Pr[f(\x) \neq h(\z)] + \tfrac34\Pr[g(\y) \neq h(\z)]\\
& = \text{avg}\left\{
			\begin{array}{r}
				\Pr[f(\x) \neq h(\z) ], \\
				\Pr[g(\y) \neq h(\z) ], \\
				\Pr[g(\y') \neq h(\z) ], \\
				\Pr[g(\y'') \neq h(\z) ]\phantom{,}
			\end{array} \right\}\\
&\leq \frac34 + \frac14 \E[\fournat(f(\x), g(\y), g(\y'), g(\y''))].
\end{align*}
This is because if $\fournat$ fails to hold on the tuple $(f(\x), g(\y), g(\y'), g(\y''))$, then $h(\z)$ can disagree with at most~$3$ of them.

At this point, we have removed $h$ from the test analysis and have uncovered what appears to be a hidden $\fournat$ test inside the $\alin{2}$ Test: simply generate four strings $\x$, $\y$, $\y'$, and $\y''$ as described earlier, and test $\fournat(f(\x), g(\y), g(\y'), g(\y''))$.  With some renaming of variables, this is exactly what our $\fournat$ Test does:

\begin{center}
\framebox{$\fournat$ Test}
\end{center}
\qquad Given folded functions $f : \Z_3^{K} \rightarrow \Z_3$, $g:\Z_3^{dK} \rightarrow \Z_3$:
\begin{itemize}
\item Let $\x \in \Z_3^K$ be uniformly random.
\item Select $\y, \z, \w$ as follows: for each $i \in [K], j \in [d]$, select $((\y[i])_j, (\z[i])_j, (\w[i])_j)$ uniformly at random from the elements of $\Z_3$ satisfying $\buddy(\x_i, (\y[i])_j,  (\z[i])_j, (\w[i])_j)$.
\item Test $\fournat(f(\x), g(\y), g(\z), g(\w))$.
\end{itemize}

\myfig{.75}{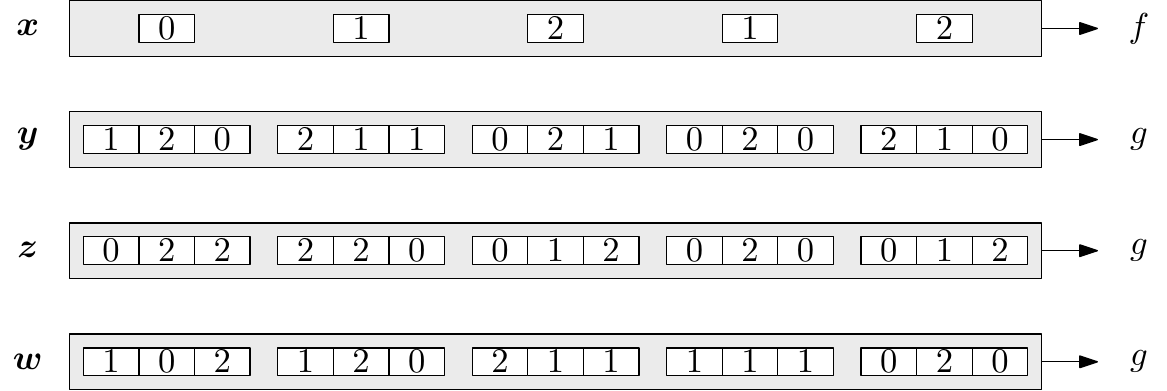}{An illustration of the $\fournat$ test distribution; $d = 3$, $K = 5$}{fig:nat-test}

Above is an illustration of this test.  In this illustration, the strings $\z$ and $\w$ were derived from the strings in Figure~\ref{fig:test} using the process detailed above for generating $\y'$ and $\y''$.  Note that each column is missing one of the elements of $\Z_3$, and that each column satisfies the $\buddy$ predicate.  Because satisfying $\buddy$ implies satisfying $\fournat$, matching dictators pass this test with probability~$1$.  On the other hand, it can be seen that nonmatching dictators pass the test with probability $\frac23$.  In the next section we show that this is optimal among functions $f$ and $g$ without ``matching influential coordinates/blocks''.

(As one additional remark, our $\alin{2}$ Test is basically the composition of the $\fournat$ Test with the gadget from Lemma~\ref{lem:fournat-alin}.  In this test, if we instead performed the $f(\x) \neq h(\z)$ test with probability $\frac13$ and the $g(\y) \neq h(\z)$ test with probability $\frac23$, then the resulting test would basically be the composition of a $\alin{3}$ test with a suitable $\alin{3}$-to-$\alin{2}$ gadget.)

\subsection{Analysis of $\fournat$ Test}\label{sec:analysis}

Let $\omega = e^{2\pi i/3}$, and set $\roots_3 = \{\omega^0, \omega^1, \omega^2\}$.  In what follows, we identify $f$ and $g$ with the functions $\omega^f$ and $\omega^g$, respectively, whose range is $\roots_3$ rather than $\Z_3$.  Set $L = dK$.  The remainder of this section is devoted to the proof of the following lemma:
\begin{lemma}\label{lem:big-fourier}
Let $f:\Z_3^K\rightarrow \roots_3$ and $g:\Z_3^{dK}\rightarrow \roots_3$.  Then
\begin{equation*}
\E[\fournat(f(\x), g(\y), g(\z), g(\w))]
\leq  \tfrac{2}{3} + \tfrac{2}{3}\sum_{\alpha \in \Z_3^L}\vert\hat{f}(\pih(\alpha))\vert\cdot \vert\hat{g}(\alpha)\vert^2\cdot(1/2)^{\#\alpha}
\end{equation*}
\end{lemma}

The first step is to ``arithmetize'' the $\fournat$ predicate.  It is not hard to verify that
\begin{align*}
  \fournat(a_1, a_2, a_3, a_4) &= \frac{5}{9} + \frac{1}{9} \sum_{i \ne j} \omega^{a_i} \overline{\omega}^{a_j} - \frac{1}{9} \sum_{i < j < k} \omega^{a_i} \omega^{a_j} \omega^{a_k} - \frac{1}{9} \sum_{i < j < k} \overline{\omega}^{a_i}\overline{\omega}^{a_j}\overline{\omega}^{a_k} \\
&= \frac59 + \frac{2}{9}\sum_{i<j}\Re[\omega^{a_i}\overline{\omega}^{a_j}] - \frac{2}{9}\sum_{i<j<k} \Re[\omega^{a_i}\omega^{a_j}\omega^{a_k}].
\end{align*}
Using the symmetry between $\y$, $\z$, and $\w$, we deduce
\begin{multline}
\E[\fournat(f(\x), g(\y), g(\z), g(\w))] \\= \tfrac59 + \tfrac23 \Re \E[f(\bx)\overline{g(\y)}] + \tfrac23 \Re \E[g(\by)\overline{g(\bz)}] - \tfrac23 \Re \E[f(\bx)g(\by)g(\bz)] - \tfrac29 \Re \E[g(\by)g(\bz)g(\bw)]. \label{eq:bigexpansion}
\end{multline}
In the second term in the RHS of~\eqref{eq:bigexpansion} we in fact have $\E[f(\bx)\overline{g(\y)}] = 0$.  This is because $\bx$ and $\by$ are independent, and hence $\E[f(\bx)\overline{g(\y)}] = \E[f(\bx)]\E[\overline{g(\y)}] = 0 \cdot 0$ since $f$ and $g$ are folded.  Regarding the third term of the RHS in~\eqref{eq:bigexpansion}, this also turns out to be~$0$ by virtue of $g$ being folded.  This can be proven using a Fourier-analytic argument; we present here an alternate combinatorial argument:
\begin{lemma}
$\E[g(\y)\overline{g(\z)}] = 0$.
\end{lemma}
\begin{proof}
Fix any value $y \in \Z_3^L$ for $\y$. Consider the function $t:\Z_3^K \times \Z_3^L \rightarrow \Z_3^K \times \Z_3^L$ defined as $t(x, z) = (x + 1, z-1)$, where all arithmetic is performed modulo 3.  Note that $t$ has order 3, meaning that $t(t(t(x, z))) = (x, z)$.  This allows us to group values for $\x$ and $\z$ into sets of size three as follows: put $(x, z) \in \Z_3^K \times \Z_3^L$ into the set $T(x, z) = \{(x, z), t(x, z), t(t(x, z))\}$.  Because $t$ is invertible and of order 3, each pair $(x, z)$ is a member of only one set: $T(x, z)$.

Conditioned on $\y = y$, if $(x, z)$ is in the support of the test, then all $(x', z') \in T(x, z)$ are also in the support of the test.  This is because the strings which are in the support of the test are exactly the strings $x$ and $z$ for which the set $\{(x_{\pi})_i, y_i, z_i\}\subseteq \Z_3$ is of size 2, for all $i \in [L]$.  These strings, in turn, are exactly those for which $x_{\pi} + y + z \not\equiv 0 \pmod{3}$.  But if $(x', z') = t(x, z)$, then
\begin{equation*}
x'_\pi + y + z' \equiv (x_\pi + 1) + y + (z - 1) \equiv x_\pi + y + z \not\equiv 0 \pmod{3}.
\end{equation*}
This shows that $t(x, z)$ is in the support of the test, conditioned on $\y = y$.  As $T(x', z') = T(x, z)$, the same holds for $t(t(x,z))$.

When conditioned on $\y = y$, each pair $(x, z)$ in the support of the test occurs with equal probability.  To see this, first note that $\x$ is pairwise independent from $\y$.  In other words, any value $x$ for $\x$ is equally likely, regardless of $y$.  Then, conditioned on $\x = x$ and $\y=y$, there are exactly two possibilities for each index of $\z$, both of which occur with half probability.  Thus, the event $(x, z)$ occurs with the same probability, no matter the values of $x$ or $z$.

Consider an arbitrary set $T(x, z)$.  Conditioned on $(\x, \z)$ falling in $T(x, z)$, the value of $(\x, \z)$ is a uniformly random element of this set.  This means that $\z$ is equally likely to be $z$, $z-1$, or $z-2$.  By the folding of $g$, $g(\z)$ is therefore equally likely to be one of $\omega^0, \omega^1$, or $\omega^2$.  As this happens for any choice of the set $T(x, z)$, $g(\z)$ is uniform on $\roots_3$, even when conditioned on $\y=y$.  Thus, $\E[g(\y)\overline{g(\z)}] = 0$ as desired.
\end{proof}

Equation~\eqref{eq:bigexpansion} has now been reduced to
\begin{equation}
\eqref{eq:bigexpansion} = \tfrac{5}{9} - \tfrac{2}{3} \Re \E[f(\x)g(\y)g(\z)] - \tfrac{2}{9}\Re \E[g(\y)g(\z)g(\w)]. \label{eq:smallexpansion}
\end{equation}
As $g(\y)g(\z)g(\w)$ is always in $\roots_3$, $\Re \E[g(\y)g(\z)g(\w)]$ is always at least $-\frac12$.  Therefore,
\begin{equation}
\eqref{eq:smallexpansion} \leq \tfrac{2}{3} - \tfrac{2}{3} \Re \E[f(\x)g(\y)g(\z)]. \label{eq:tinyexpansion}
\end{equation}
It remains to handle the $\E[f(\x)g(\y)g(\z)]$ term, which is the subject of our next lemma.  This is done through a standard argument in the style of \Hastad~\cite{Has01}.
\begin{lemma}
$\E[f(\x)g(\y)g(\z)] = \sum_{\alpha \in \Z_3^L}\hat{f}(\pih(\alpha))\hat{g}(\alpha)^2\left(-\frac{1}{2}\right)^{\#\alpha}$.
\end{lemma}
\begin{proof}
Begin by expanding out $\E[f(\x)g(\y)g(\z)]$:
\begin{equation}
\E[f(\x)g(\y)g(\z)] = \sum_{\substack{
                                               \alpha \in \Z_3^K, \beta, \gamma \in \Z_3^L\\
                                                \vert \alpha \vert \equiv \vert \beta \vert \equiv \vert \gamma \vert \equiv 1 \pmod{3}}}
\hat{f}(\alpha)\hat{g}(\beta)\hat{g}(\gamma)\E[\chi_\alpha(\x)\chi_{\beta}(\y)\chi_\gamma(\z)].\label{eq:fourierexpanded}
\end{equation}
We focus on the products of the Fourier characters:
\begin{equation}
\E[\chi_\alpha(\x)\chi_\beta(\y)\chi_\gamma(\z)]
=  \prod_{i \in [K]}\E[\chi_{\alpha_i}(\x_i) \chi_{\beta[i]}(\y[i])\chi_{\gamma[i]}(\z[i])]\label{eq:fcharacters}
\end{equation}
We can attend to each block separately:
\begin{align}
\E[\chi_{\alpha_i}(\x_i) \chi_{\beta[i]}(\y[i])\chi_{\gamma[i]}(\z[i])]
= &\E\left[\omega^{\alpha_i \cdot \x_i + \beta[i] \cdot \y[i] + \gamma[i] \cdot \z[i]}\right]\nonumber\\
= &\E_{\x}\left[\omega^{\alpha_i \cdot a}\prod_{j:\pi(j) = i}\underbrace{\E_{\y, \z}\left[\omega^{\beta_j \y_j + \gamma_j \z_j} \mid \x_i = a\right]}_{(*)}\right].\label{eq:blockcharacters}
\end{align}

Now, consider the expectation $(*)$.  The distribution on the values for $(\y_j, \z_j)$ is uniform on the six possibilities $(a+1, a+1)$, $(a+2, a+2)$, $(a, a+1)$, $(a, a+2)$, $(a+1, a)$, and $(a+2, a)$.  We claim that $(*)$ is nonzero if and only if $\beta_j \equiv \gamma_j \pmod{3}$.  If, on the other hand, $\beta_j \not\equiv \gamma_j \pmod{3}$, then either only one of $\beta_j$ or $\gamma_j$ is zero, or neither is zero, and $-\beta_j \equiv \gamma_j \pmod{3}$.  In the first case, the expectation is either $\E[\omega^{\beta_j\y_j}\mid \x_i = a]$ or $\E[\omega^{\gamma_j\z_j}\mid \x_i = a]$ for a nonzero $\beta_j$ or a nonzero $\gamma_j$, respectively.  Both of these expectations are zero, as both $\y_j$ and $\z_j$ are uniform on $\Z_3$.  In the second case,
\begin{align*}
\E[\omega^{\beta_j \y_j + \gamma_j \z_j} \mid \x_i = a] = &\E[\omega^{\beta_j\y_j - \beta_j\z_j} \mid \x_i = a]\\
= &\E[\omega^{\beta_j(\y_j - \z_j)} \mid \x_i = a],
\end{align*}
which is zero, because $\beta_j$ is nonzero, and $\y_j - \z_j$ is uniformly distributed on $\Z_3$.

Thus, when $(*)$ and Equation~\eqref{eq:fcharacters} are nonzero, $\beta \equiv \gamma \pmod{3}$.  This means that $(*) = \E[\omega^{\beta_j(\y_j+\z_j)}\mid \x_i = a]$.  When $\beta_j = 0$, this is clearly $1$.  Otherwise, as either $\y_j + \z_j \equiv 2a + 1 \pmod{3}$ or $\y_j + \z_j \equiv 2a+2 \pmod{3}$, each with probability half,  this is equal to
\begin{equation*}
(*) = \frac{1}{2}\left(\omega^{\beta_j(2a+1)} + \omega^{\beta_j(2a+2)}\right)
 = \frac{\omega^{2 a \beta_j}}{2}(\omega^1 + \omega^2) = -\frac{\omega^{2a\beta_j}}{2}.
\end{equation*}
In summary, when $\beta= \gamma$, $(*) = \left(-\frac{1}{2}\right)^{\#\beta_j}\omega^{2a\beta_j}$.

We can now rewrite Equation~\eqref{eq:blockcharacters} as
\begin{equation*}
\eqref{eq:blockcharacters}
=  \E_{\x}\left[\omega^{\alpha_i \cdot a} \prod_{j:\pi(j) = i} \left(-\frac{1}{2}\right)^{\#\beta_j}\omega^{2a\beta_j}\right]
=  \E_{\x}\left[\left(-\frac{1}{2}\right)^{\#\beta[i]}\omega^{(\alpha_i  + 2 \vert \beta[i]\vert)a}\right].
\end{equation*}
Note that the exponent of $\omega$, $(\alpha_i +2\vert \beta[i]\vert) a$, is zero if $\alpha_i \equiv \vert \beta[i]\vert \pmod{3}$, in which case the expectation is just the constant $(-1/2)^{\#\beta[i]}$.  This occurs for all $i \in [K]$ exactly when $\alpha = \pih(\beta)$.  If, on the other hand, $\alpha_i + 2\vert\beta[i]\vert$ is nonzero, then the entire expectation is zero because $a$, the value of $\x_i$, is uniformly random from $\Z_3$.  Thus, Equation \eqref{eq:fcharacters} is nonzero only when $\alpha = \pih(\beta)$ and $\beta = \gamma$, in which case it equals
\begin{equation*}
\eqref{eq:fcharacters} = \left(-\frac{1}{2}\right)^{\#\beta}.
\end{equation*}
We may therefore conclude with
\begin{equation*}
\eqref{eq:fourierexpanded} = \sum_{\alpha \in \Z_3^L}\hat{f}(\pih(\alpha))\hat{g}(\alpha)^2\left(-\frac{1}{2}\right)^{\#\alpha}. \qedhere
\end{equation*}
\end{proof}

Substituting this result into~\eqref{eq:tinyexpansion} yields
\begin{align*}
\E[\fournat(f(\x), g(\y), g(\z), g(\w))]
&\leq \tfrac{2}{3} - \tfrac{2}{3} \Re \sum_{\alpha \in \Z_3^L}\hat{f}(\pih(\alpha))\hat{g}(\alpha)^2\left(-\frac{1}{2}\right)^{\#\alpha}\\ &\leq \tfrac{2}{3} + \tfrac{2}{3}\sum_{\alpha \in \Z_3^L}\vert\hat{f}(\pih(\alpha))\vert\cdot \vert\hat{g}(\alpha)\vert^2\cdot(1/2)^{\#\alpha},
\end{align*}
completing the proof of Lemma~\ref{lem:big-fourier}.


\section{Hardness of $\fournat$}\label{app:fournat}

In this section, we show the following theorem:
\begin{named}{Theorem \ref{thm:fournat-hardness} (detailed)}
For all $\eps > 0$, it is $\NP$-hard to $(1, \frac23 + \eps)$-decide the $\fournat$ problem.  In fact, in the ``yes case'', all $\fournat$ constraints can be satisfied by $\twopair$ assignments.
\end{named}
Combining this with Lemma~\ref{lem:fournat-alin} yields Theorem~\ref{thm:alin-hardness}, and combining this with Corollary~\ref{cor:fournat-twotoone} yields Theorem~\ref{thm:twotoone-hardness}.  It is not clear whether this gives optimal hardness assuming perfect completeness.  The $\fournat$ predicate is satisfied by a uniformly random input with probability $\frac59$, and by the method of conditional expectation this gives a deterministic algorithm which $(1, \frac59)$-approximates the $\fournat$ CSP.  This leaves a gap of $\frac19$ in the soundness, and to our knowledge there are no better known algorithms.

On the hardness side, consider a uniformly random satisfying assignment to the $\buddy$ predicate.  It is easy to see that each of the four variables is assigned a uniformly random value from $\Z_3$, and also that the variables are pairwise independent.  As any satisfying assignment to the $\buddy$ predicate also satisfies the $\fournat$ predicate, the work of Austrin and Mossel~\cite{AM09} immediately implies that $(1-\eps, \frac59+\eps)$-approximating the $\fournat$ problem is $\NP$-hard under the Unique Games conjecture.  Thus, if we are willing to sacrifice a small amount in the completeness, we can improve the soundness parameter in Theorem~\ref{thm:fournat-hardness}.  Whether we can improve upon the soundness without sacrificing perfect completeness is open.

We now arrive at the proof of Theorem~\ref{thm:fournat-hardness}.  The proof is entirely standard, and proceeds by reduction from $d$-to-1 Label Cover.  It makes use of our analysis of the $\fournat$ Test, which is presented in Appendix~\ref{sec:analysis}.  One preparatory note: most of the proof concerns functions $f:\Z_3^K\rightarrow \Z_3$ and $g:\Z_3^{dK}\rightarrow \Z_3$.  However, we also be making use of Fourier analytic notions defined in Section~\ref{sec:fourier}, and this requires dealing with functions whose range is $\roots_3$ rather than $\Z_3$.  Thus, we associate $f$ and $g$ with the functions $\omega^f$ and $\omega^g$, and whenever Fourier analysis is used it will actually be with respect to the latter two functions.

\begin{proof}
Let $G = (U \cup V, E)$ be a $d$-to-1 Label Cover instance with alphabet size $K$ and $d$-to-1 maps $\pi_e:[dK]\rightarrow [K]$ for each edge $e \in E$.  We construct a $\fournat$ instance by replacing each vertex in $G$ with its Long Code and placing constraints on adjacent Long Codes corresponding to the tests made in the $\fournat$ Test.  Thus, each $u\in U$ is replaced by a copy of the hypercube $\Z_3^K$ and labeled by the function $f_u:\Z_3^K \rightarrow \Z_3$.  Similarly, each $v \in V$ is replaced by a copy of the Boolean hypercube $\Z_3^{dK}$ and labeled by the function $g_v:\Z_3^{dK} \rightarrow \Z_3$.  Finally, for each edge $\{u, v\} \in E$, a set of $\fournat$ constraints is placed between $f_u$ and $g_v$ corresponding to the constraints made in the $\fournat$ Test, and given a weight equal to the probability the constraint is tested in the $\fournat$ Test multiplied by the weight of $\{u,v\}$ in $G$.  This produces a $\fournat$ instance whose weights sum to 1 which is equivalent to the following test:
\begin{itemize}
\item Pick an edge $e=(u, v) \in E$ uniformly at random.
\item Reorder the indices of $g_v$ so that the $k$th group of $d$ indices corresponds to $\pi_e^{-1}(k)$.
\item Run the $\fournat$ test on $f_u$ and $g_v$.  Accept iff it does.
\end{itemize}

\paragraph{Completeness}
If the original Label Cover instance is fully satisfiable, then there is a function $F:U\cup V \rightarrow [dK]$ for which $\mathsf{val}(F)=1$.  Set each $f_u$ to the dictator assignment $f_u(x) = x_{F(u)}$ and each $g_v$ to the dictator assignment $g_v(y) = y_{F(v)}$.  Let $e=\{u, v\} \in E$.  Because $F$ satisfies the constraint $\pi_e$, $F(u) = \pi_e(F(v))$.  Thus, $f_u$ and $g_v$ correspond to ``matching dictator'' assignments, and above we saw that matching dictators pass the $\fournat$ Test with probability 1.  As this applies to every edge in $E$, the $\fournat$ instance is fully satisfiable.

\paragraph{Soundness}
Assume that there are functions $\{f_u\}_{u \in U}$ and $\{g_v\}_{v \in V}$ which satisfy at least a $\frac23+\eps$ fraction of the $\fournat$ constraints.  Then there is at least an $\eps/2$ fraction of the edges $e =\{u, v\} \in E$ for which $f_u$ and $g_v$ pass the $\fournat$ Test with probability at least $\frac23+\eps/2$.  This is because otherwise the fraction of $\fournat$ constraint satisfied would be at most
\begin{equation*}
\left(1-\frac{\eps}{2}\right)\left(\frac{2}{3}+\frac{\eps}{2}\right) + \frac{\eps}{2}(1)
= \frac{2}{3} + \frac{2\eps}{3} - \frac{\eps^2}{4} < \frac{2}{3} + \eps.
\end{equation*}
Let $E'$ be the set of such edges, and consider $\{u, v\} \in E'$.  Set $L = dK$.
By Lemma~\ref{lem:big-fourier},
\begin{equation*}
\frac{2}{3} + \frac{\eps}{2} \leq \Pr[\text{$f_u$ and $g_v$ pass the $\fournat$ test}]
\leq \frac{2}{3} + \frac{2}{3}\left(\sum_{\alpha \in \Z_3^L}\left\vert\hat{f}_u(\pih(\alpha))\right\vert\left\vert\hat{g}_v(\alpha)\right\vert^2\left(\frac{1}{2}\right)^{\#\alpha}\right),
\end{equation*}
meaning that
\begin{equation}
\frac{3\eps}{4}
\leq \sum_{\alpha \in \Z_3^L}\left\vert\hat{f}_u(\pih(\alpha))\right\vert\left\vert\hat{g}_v(\alpha)\right\vert^2\left(\frac{1}{2}\right)^{\#\alpha}.\label{eq:fourier-preprob}
\end{equation}
Parseval's equation tells us that $\sum_{\alpha \in \Z_3^L}\vert \hat{g}_v(\alpha) \vert^2 = 1$.  The function $\hat{g}_v$ therefore induces a probability distribution on the elements of $\Z_3^L$.  As a result, we can rewrite Equation~\eqref{eq:fourier-preprob} as
\begin{equation}
\frac{3\eps}{4}\leq \E_{\alpha \sim \hat{g}_v}\left[\left\vert\hat{f}_u(\pih(\alpha))\right\vert\left(\frac{1}{2}\right)^{\#\alpha}\right].\label{eq:fourier-prob}
\end{equation}
As previously noted, $\vert \hat{f}_u(\pih(\alpha))\vert$ is less than 1 for all $\alpha$, so the expression in this expectation as never greater than 1.  We can thus conclude that
\begin{equation*}
\frac{3\eps}{8}\leq \Pr_{\alpha \sim \hat{g}_v}\underbrace{\left[\left\vert \hat{f}_u(\pih(\alpha))\right\vert \left(\frac{1}{2}\right)^{\#\alpha} \geq \frac{3\eps}{8}\right]}_{\mathsf{GOOD}_\alpha},
\end{equation*}
as otherwise the expectation in Equation~\eqref{eq:fourier-prob} would be less than $3\eps/4$.  Call the event in the probability $\mathsf{GOOD}_\alpha$.  When $\mathsf{GOOD}_\alpha$ occurs, the following happens:
\begin{itemize}
\item $\vert \hat{f}_u(\pih(\alpha))\vert^2 \geq 9\eps^2/64$.
\item $\#\alpha \leq \log_2(8/3\eps)$.  Furthermore, as $f_u$ is folded, $\#\alpha > 0$.
\end{itemize}
This suggests the following randomized decoding procedure for each $u \in U$: pick an element $\beta \in \Z_3^K$ with probability $\vert \hat{f}_u(\beta)\vert^2$ and choose one of its nonzero coordinates uniformly at random.  Similarly, for each $v \in V$, pick an element $\alpha \in \Z_3^L$ with probability $\vert \hat{g}_v(\alpha)\vert^2$ and choose one of its nonzero coordinates uniformly at random.  In both cases, nonzero coordinates are guaranteed to exist because all the $f_u$'s and $g_v$'s are folded.

Now we analyze how well this decoding scheme performs for the edges $e = \{u, v\} \in E'$ (we may assume the other edges are unsatisfied).  Suppose that when the elements of $\Z_3^K$ and $\Z_3^L$ were randomly chosen, $g_v$'s set $\alpha$ was in $\mathsf{Good}_\alpha$, and $f_u$'s set $\beta$ equals $\pih(\alpha)$.  Then, as $\#\alpha \leq \log_2(8/3\eps)$, and each label in $\pih(\alpha)$ has at least one label in $\alpha$ which maps to it, the probability that matching labels are drawn is at least $1/\log_2(8/3\eps)$.  Next, the probability that such an $\alpha$ and $\beta$ are drawn is
\begin{equation*}
\sum_{\alpha \in \mathsf{GOOD}}\vert \hat{f}_u(\pih(\alpha)) \vert^2\vert \hat{g}_v(\alpha)\vert^2
\geq \frac{9\eps^2}{64}\sum_{\alpha \in \mathsf{GOOD}}\vert \hat{g}_v(\alpha)\vert^2
\geq \frac{9\eps^2}{64}\frac{3\eps}{8} = \frac{27\eps^3}{512}.
\end{equation*}
Combining these, the probability that this edge is satisfied is at least $27\eps^3/512\log_2(8/3\eps)$.  Thus, the decoding scheme satisfies at least
\begin{equation*}
\frac{27\eps^3}{512\log_2(8/3\eps)}\cdot \frac{\vert E'\vert}{\vert E\vert} \geq \frac{27\eps^4}{1024\log_2(8/3\eps)}
\end{equation*}
fraction of the Label Cover edges in expectation.  By the probabilistic method, an assignment to the Label Cover instance must therefore exist which satisfies at least this fraction of the edges.

We now apply Theorem~\ref{thm:mosraz}, setting the soundness value in that theorem equal to $O(\eps^5)$, which concludes the proof.  
\end{proof}

\bibliographystyle{alpha}
\bibliography{../bib/odonnell-bib}

\end{document}